\documentclass[conference,letterpaper]{IEEEtran}

\usepackage[utf8]{inputenc} 
\usepackage[T1]{fontenc}
\usepackage{url}
\usepackage{ifthen}
\usepackage{cite}
\usepackage[cmex10]{amsmath} % [cmex10] option to ensure complicance with IEEE Xplore
\usepackage{amsthm,amssymb}
\usepackage{proof}
\newtheorem{theorem}{Theorem}
\newtheorem{lemma}{Lemma}
\newtheorem{proposition}{Proposition}

\newtheorem{property}{Property}
\theoremstyle{definition}
\newtheorem{definition}{Definition}
\theoremstyle{remark}
\newtheorem{remark}{Remark}

\newcommand\E{\mathbb{E}}
\renewcommand\P{\mathbb{P}}

\newcommand{\T}{\mathbf{T}}
\newcommand{\X}{\mathbf{X}}
\newcommand{\Y}{\mathbf{Y}}
\newcommand{\N}{\mathbf{N}}

\usepackage{graphicx,subfigure}
\DeclareGraphicsExtensions{.pdf,.png,.jpg}
\usepackage{enumerate}
\usepackage{tikz}
\usetikzlibrary{shapes,arrows}
\usetikzlibrary{positioning,fit,calc}
\tikzstyle{block} = [draw, fill=white, rectangle, 
    minimum height=3em, minimum width=6em]
\tikzstyle{input} = [coordinate]
\tikzstyle{output} = [coordinate]

\begin{document}
\title{On Conditional $\alpha$-Information\\and its Application to Side-Channel Analysis} 

\author{%
	\IEEEauthorblockN{Yi Liu\IEEEauthorrefmark{1}, Wei Cheng\IEEEauthorrefmark{1}, Sylvain Guilley\IEEEauthorrefmark{2}\IEEEauthorrefmark{1}, and Olivier Rioul\IEEEauthorrefmark{1}}\\[-2mm]
	\IEEEauthorblockA{\IEEEauthorrefmark{1}%
		LTCI, T\'el\'ecom Paris,
		Polytechnique de Paris,
		91120, Palaiseau, France,
		firstname.lastname@telecom-paris.fr}
	\IEEEauthorblockA{\IEEEauthorrefmark{2}%
		Secure-IC S.A.S.,
		75015, Paris, France,
		sylvain.guilley@secure-ic.com}
}

\maketitle

%% If your paper is eligible for the student paper award, please add
%% the comment "THIS PAPER IS ELIGIBLE FOR THE STUDENT PAPER
%% AWARD." as a first line in the abstract. 
\noindent\textmd{\small THIS PAPER IS ELIGIBLE FOR THE STUDENT PAPER AWARD}
%% For the final version of the accepted paper, please do not forget to remove this !
\medskip

\begin{abstract}
A conditional version of Sibson's $\alpha$-information is defined using a simple closed-form ``log-expectation'' expression, which satisfies important properties such as consistency, uniform expansion, and data processing inequalities. This definition is compared to previous ones, which in contrast do not satisfy all of these properties.
Based on our proposal and on a generalized Fano inequality, 
we extend the case $\alpha=1$ of previous works to obtain sharp universal upper bounds for the probability of success of any type side-channel attack, particularly when $\alpha=2$.
\end{abstract}

%\textit{A full version of this paper is accessible at:}
%\url{https://arxiv.org/pdf/21xx.xxxx.pdf} 

% 5 pages + references

\section{Introduction}

Mutual information as a theoretical tool to analyse the capability of an attacker to perform side-channel analysis has been advocated since Standaert et al.~\cite{StandaertMalkinYung09}. 
The communication channel model for this problem was used in~\cite{HeuserRioulGuilley14} to optimize the side-channel attack distinguishers for any given type of leakage model.
Recently, Chérisey et al.~\cite{CheriseyGuilleyRioulPiantanida19c,CheriseyGuilleyRioulPiantanida19}  used such information-theoretic tools to establish some \emph{universal} inequalities between the probability of success of a side-channel attack and the minimum number of queries to reach a given success rate. Such inequalities are ``universal'' in the sense that they can apply to any type of attack and depend only on the leakage model.

In this paper, we aim at extending the approach of~\cite{CheriseyGuilleyRioulPiantanida19c,CheriseyGuilleyRioulPiantanida19} to R\'enyi information quantities depending on a parameter $\alpha$.
For that we need the following ingredients that were crucial in the derivation steps of\cite{CheriseyGuilleyRioulPiantanida19c,CheriseyGuilleyRioulPiantanida19}:
\begin{itemize}
\item a \emph{closed-form expression} of conditional mutual information,  amenable to efficient numerical estimation;
\item a \emph{data processing inequality} of conditional mutual information over a ``conditional'' Markov chain for a given plain or cypher text $T$ (known to the attacker);
\item a \emph{expansion property} of conditional mutual information, i.e., its decomposition as a difference between conditional entropies, valid at least when the secret is assumed \emph{uniformly distributed};
\item a \emph{Fano inequality} which yields a lower bound on mutual information that depends on the probability of success (or equivalently on the probability of error).
\end{itemize}
Our aim, therefore, is to establish all of these properties for a suitably defined conditional R\'enyi version of mutual information of order $\alpha>0$.

The rest of this paper is organized as follows. Section~\ref{sec-back} reviews some useful definitions and properties of R\'enyi informational quantities. Emphasis is made on consistency, uniform expansion and data processing inequalities. Section~\ref{sec-def} then proposes a natural definition of conditional $\alpha$-information satisfying the required properties and Section~\ref{sec-comp} makes a detailed comparison to previous proposals. Section~\ref{sec-sca} presents the main result applied to side-channel analysis, which is then validated using simulations.

\section{Background and Definitions}\label{sec-back}

\subsection{$\alpha$-Entropy and $\alpha$-Divergence}

Rényi entropy and divergence are well-known generalizations of Shannon's entropy and Kullback-Leibler divergence:
\begin{definition}
Assume that either $0<\alpha<1$ or $1<\alpha<+\infty$ (the limiting values $0,1,+\infty$ 
being obtained by taking limits).

The \emph{$\alpha$-entropy} of a probability distribution $P$ and \emph{$\alpha$-divergence} of $P$ from $Q$ are defined as
\begin{align}
H_\alpha(P)&= \tfrac{\alpha}{1-\alpha} \log\|p\|_\alpha\\
D_\alpha (P\|Q)&=  \tfrac{1}{\alpha-1} \log\langle p\|q\rangle^\alpha_\alpha
\end{align}
where we have used the special notation:
\begin{align}
\|p\|_\alpha &= \bigl(\int p^\alpha d\mu \bigr)^{1/\alpha}\\
\langle p\|q\rangle_\alpha &= \bigl(\smash{\int} p^\alpha q^{1-\alpha} d\mu \bigr)^{1/\alpha}
\end{align}
with the following convention: All considered probability distributions $P,Q$ possess a dominating measure $\mu$ such that $P \ll \mu$ and $Q \ll \mu$, the corresponding lower-case letters $p,q$ are densities of $P,Q$ with respect to $\mu$.
\end{definition}

\begin{remark}
When $\mu$ is a counting measure we obtain the classical definitions for discrete random variables; when $\mu$ is the Lebesgue measure we obtain the corresponding definitions for continuous variables. 
While it is easily seen that the definition of $\alpha$-divergence does not depend on the chosen dominating measure $\mu$, that of $\alpha$-entropy does. 
\end{remark}

A link between these two quantities is the following \emph{uniform expansion property} (UEP).
Let $U\sim\mathcal{U}(M)$ be uniformly distributed over a set of finite $\mu$-measure $M$. (In the discrete case $U$ simply takes $M$ equiprobable values.) Since $u\equiv \frac{1}{M}$ we have $\langle p\|u\rangle_\alpha=M^{\frac{\alpha-1}{\alpha}}\|p\|_\alpha$, hence
\begin{property}[UEP of $\alpha$-Divergence]
$D_\alpha(P\|U) =H_\alpha(U)-H_\alpha(P)=\log M - H_\alpha(P)$. 
\end{property}

Another important property is the \emph{data processing inequality} (DPI).
A random transformation given by a conditional distribution $P_{Y|X}$ is noted $P_X\, \to \boxed{P_{Y|X}}\, \to P_Y$ if a random variable $X\sim P_X$ is input and the output distribution $P_Y$ satisfies $p_Y(y)=\int p_{Y|X}(y|x)p_X(x)\,d \mu(x)$. Similarly for $Q_X\, \to \boxed{P_{Y|X}}\, \to Q_Y$ we have $q_Y(y)=\int_{\mathcal{X}}p(y|x)q_X(x)\,d\mu(x)$.
\begin{property}[DPI for $\alpha$-Divergence~\cite{PolyanskiyVerdu10,Rioul21}] 
Any transformation can only reduce $\alpha$-divergence: $D_\alpha(P_X\|Q_X)\geq D_\alpha(P_Y\|Q_Y)$.
\label{dpi_alpha}
\end{property}

\subsection{Conditional $\alpha$-Entropy and $\alpha$-Divergence}

Both definitions of $\alpha$-divergence and $\alpha$-entropy have been extended to \emph{conditional} versions, in a fairly natural way:
\begin{definition}\label{def-conddiv}
The conditional $\alpha$-divergence is defined as~\cite{Verdu15}  
\begin{equation}
D_\alpha(P_{Y|X}\|Q_{Y|X}|P_X)= D_\alpha(P_{Y|X}P_X\|Q_{Y|X}P_X)
\end{equation}
\end{definition}
This definition is consistent with the unconditional one:
\begin{property}[Consistency of Conditional $\alpha$-Divergence w.r.t. $\alpha$-Divergence]
If $X\equiv 0$ then $D_\alpha(P_{Y|X}\|Q_{Y|X}|P_X)=D_\alpha(P_Y\|Q_Y)$.
\end{property}
\noindent Here following Shannon~\cite{Shannon53}
we have noted $X\equiv 0$ for any random variable independent of everything else  considered (e.g., a constant variable).

In Definition~\ref{def-conddiv} we remark that the \emph{expectation} over the conditioned variable is only taken \emph{inside the logarithm} in the $\alpha$-divergence's expression:
\begin{equation}
D_\alpha(P_{Y|X}\|Q_{Y|X}|P_X)=\tfrac{1}{\alpha-1} \log \E_{X}  \langle p_{Y|X}\|q_{Y|X}\rangle^\alpha_\alpha
\end{equation}
A similar ``log-expectation'' definition holds for the following preferred form of the conditional $\alpha$-entropy (a.k.a. Arimoto's conditional entropy). Considering the expression $H_\alpha(X)=H_\alpha(P_X)= \frac{\alpha}{1-\alpha} \log\|p_X\|_\alpha$ and taking the expectation over a conditioned variable inside the logarithm yields the following
\begin{definition}
The conditional $\alpha$-entropy of $X$ given $Y$ is defined as~\cite{Arimoto75,FehrBerens14}   
\begin{equation}
H_\alpha(X|Y) = \frac{\alpha}{1-\alpha} \log \E_Y \|p_{X|Y}\|_\alpha
\end{equation}
\end{definition}
%\begin{remark}
%Further conditioning is again made inside the logarithm,  e. g.,
%$
%H_\alpha(X|YZ) = \frac{\alpha}{1-\alpha} \log \E_{Z}\E_{Y|Z} \|p_{X|YZ}\|_\alpha 
%$.
%\end{remark}

Among other variations of conditional $\alpha$-entropy~\cite{FehrBerens14} 
it is this definition that enjoys all three important properties: \emph{consistency, UEP and DPI}.
Consistency is obvious from the definition:
\begin{property}[Consistency of Conditional $\alpha$-Entropy w.r.t. $\alpha$-Entropy]
If $Y\equiv 0$ then $H_\alpha(X|Y)= H_\alpha(X)$.
\end{property}
As in the case of the $\alpha$-entropy, since $\langle p_{X|Y}\|u\rangle_\alpha=M^{\frac{\alpha-1}{\alpha}}\|p_{X|Y}\|_\alpha$, we have the following
\begin{property}[UEP]
If $U\sim\mathcal{U}(M)$ is uniform independent ot $X$,
$D_\alpha(P_{Y|X}\|U|P_X)=H_\alpha(U)-H_\alpha(Y|X)=\log M-H_\alpha(Y|X)$.
\end{property}
\begin{property}[DPI for Conditional $\alpha$-Entropy~\cite{FehrBerens14,Rioul21}]
If $X-Y-Z$ forms a Markov chain, then $H_\alpha(X|Y) \leq H_\alpha(X|Z)$.
\end{property}
\noindent
In particular for $Z\!\equiv\!0$, conditioning reduces $\alpha$-entropy: $H_\alpha(X|Y) \leq  H_\alpha(X|0) = H_\alpha(X)$.
More generally one has~\cite{FehrBerens14}
$H_\alpha(X|YY') \leq H_\alpha(X|Y')$.

\subsection{$\alpha$-Information}

Sibson's $\alpha$-information is perhaps the preferred generalization of Fano's classical mutual information and has found various applications~\cite{PolyanskiyVerdu10,Verdu15,TomamichelHayashi18,Rioul21,EspositoWuGastpar21,EspositoGastparIssa}: 
\begin{definition}\label{def-info}
The $\alpha$-information~\cite{Sibson69,Verdu15}
of $X$ from $Y$ is defined as
\begin{equation}\label{eq-sibson}
I_\alpha(X;Y)= \frac{\alpha}{\alpha-1} \log \E_Y \langle p_{X|Y}\|p_X\rangle_\alpha
\end{equation}
\end{definition}
\noindent This is again a ``log-expectation''  expression where one takes the expectation over $Y$ inside the logarithm in the expression of the divergence
\begin{equation*}
 D_\alpha(P_{X|Y=y}\|P_X)= \frac{\alpha}{\alpha-1} \log \langle p_{X|Y=y}\|p_X\rangle_\alpha
\end{equation*}
\begin{remark}
This construction  focuses on the distribution of $X$, conditioned on $Y$ or not. In contrast to the classical case $\alpha=1$, the resulting definition of information is not symmetric: $I_\alpha(X;Y)\ne I_\alpha(Y,X)$. Therefore, $\alpha$-information is no longer ``mutual'' when $\alpha\ne 1$.
\end{remark}
As in the case of the conditional $\alpha$-entropy, since $\langle p_{U|Y}\|u\rangle_\alpha=M^{\frac{\alpha-1}{\alpha}}\|p_{U|Y}\|_\alpha$, we have the following
\begin{property}[UEP for $\alpha$-Information~\cite{ErvenHarremoes14,Rioul21}] 
If $U\sim\mathcal{U}(M)$ is uniformly distributed, then $I_\alpha(U;Y) = H_\alpha(U) - H_\alpha(U|Y)=\log M - H_\alpha(U|Y)$.
\end{property}
\begin{property}[DPI for $\alpha$-Information~\cite{PolyanskiyVerdu10,Rioul21}]
If $W-X-Y-Z$ forms a Markov chain, then $I_\alpha(X;Y) \geq I_\alpha(W;Z)$.
\label{prop-DPI-info}
\end{property}
\begin{proof}[Proof (for completeness)]
Let $P_{X,Y}\!\to\!\boxed{P_{X,Z|X,Y}}\!\to\!P_{X,Z}\to\!\boxed{P_{W,Z|X,Z}}\!\to P_{W,Z}$. By the Markov condition, one has $P_{X,Z|X,Y}=P_{X|X} P_{Z|X,Y}=P_{X|X} P_{Z|Y}$ where $P_{X|X}$ is the identity operator; similarly $P_{W,Z|X,Z}=P_{W|X,Z}P_{Z|Z}=P_{W|X}P_{Z|Z}$. Thus if $Q_{Y} \to\!\boxed{P_{Z|Y}}\!\to Q_{Z}$, 
we find $P_{X}Q_{Y} \to\!\boxed{P_{X,Z|X,Y}}\!\to P_{X}Q_{Z}
\to\!\boxed{P_{W,Z|X,Z}}\!\to P_{W}Q_Z$.
Now by the data processing inequality for $\alpha$-divergence (Property~\ref{dpi_alpha}),
$D_\alpha(P_{X,Y}\|P_XQ_Y)\geq D_\alpha(P_{W,Z}\|P_WQ_Z)\geq I_\alpha(W;Z)$.
Minimizing over $Q_Y$ gives the announced DPI.
\end{proof}

\begin{remark}
Because of the non-symmetric nature of $\alpha$-information, the DPI corresponds to two separate statements of pre- and post-processing inequalities~\cite{PolyanskiyVerdu10}. 

We remark that the Lapidoth-Pfister \emph{mutual} information, which is symmetric, $J_{\alpha}(X;Y)=J_{\alpha}(Y;X)$ does also enjoy data processing inequalities but unfortunately does not seem to possess a closed-form expression~\cite{LapidothPfister16}.
\end{remark}

\subsection{Sibson's identity}

An important property of $\alpha$-information is \emph{Sibson's identity}.
It is straightforward to compute
\begin{align}
\langle p_{XY} \| p_X q_Y \rangle^\alpha_\alpha &=
\iint p^\alpha_{Y} p^\alpha_{X|Y} p_X^{1-\alpha} q_Y^{1-\alpha} \\
&= \bigl\langle p_Y \langle p_{X|Y}\|p_X\rangle_{\!\alpha} \,\|\, q_Y \bigr\rangle^\alpha_\alpha.
\end{align}
Defining the (suitably normalized) distribution  $q^*_Y=p_Y \langle p_{X|Y}\|p_X\rangle_{\!\alpha}  / \E_Y \langle p_{X|Y}\|p_X\rangle_{\!\alpha} $, substituting and taking the logarithm gives the following
\begin{proposition}[Sibson's identity~\cite{Sibson69,Verdu15}] 
One has
\begin{equation}\label{eq-sibson-identity}
D_\alpha(P_{XY}\|P_XQ_Y) = D(Q^*_Y\|Q_Y) + I_\alpha(X;Y), 
\end{equation}
hence the following alternate minimizing definition:
\begin{equation}\label{eq-info-min}
I_\alpha(X;Y)  =\min_{Q_Y}  D_\alpha(P_{XY}\|P_XQ_Y).
\end{equation}
\end{proposition}

\subsection{Generalized Fano's Inequality}

Assume $X$ is discrete and estimated from $Y$ using the MAP rule, with (maximal) probability of success $\P_s=\P_s(X|Y)=\E\sup_x p_{X|Y}(x|Y)$. Also let $\P_s(X)=\sup p_X$ be the probability of success when guessing $X$ without even knowing $Y$. Using the DPI for $\alpha$-information and $\alpha$-divergence, we have the following
\begin{lemma}[Rioul's Generalized Fano Inequality~{\cite[Thm.~1]{Rioul21}}]\label{lem-gen-fano}
\begin{equation}
I_\alpha(X;Y) \geq d_\alpha\bigl(\P_s(X|Y)\|\P_s(X)\bigr) 
\end{equation}
where 
\begin{equation}\label{eq-binary-div}
 d_\alpha (p\|q)=  \tfrac{1}{\alpha-1} \log \bigl( p^\alpha q^{1-\alpha} + (1-p)^\alpha (1-q)^{1-\alpha} \bigr)
\end{equation}
denotes binary $\alpha$-divergence.
\end{lemma}

\section{Conditional $\alpha$-Information}\label{sec-def}

\subsection{Definition as a Log-Expectation Expression}
As a natural continuation of the definitions in the preceding section, we  define the conditional $\alpha$-information with a ``log-expectation'' closed-form expression,  obtained by taking the expectation over the conditional variable inside the logarithm in the expression of Sibson's (unconditional) $\alpha$-information~\eqref{eq-sibson}:
\begin{definition}[Closed-Form Definition of $\alpha$-Information]\label{def-cond-info}
The $\alpha$-information~\cite{Sibson69,Verdu15} 
\begin{align}
I_\alpha(X;Y|Z)= \frac{\alpha}{\alpha-1} \log \E_Z\E_{Y|Z} \langle p_{X|YZ}\|p_{X|Z}\rangle_\alpha \notag\\
= \frac{\alpha}{\alpha-1} \log \E_{YZ} \langle p_{X|YZ}\|p_{X|Z}\rangle_\alpha
\label{eq-cond-info}
\end{align} 
To the best of our knowledge, this definition has not been considered elsewhere.
\end{definition}

\subsection{Basic Properties}
Our definition enjoys three important properties: \emph{consistency, UEP and DPI}.

\begin{property}[Consistency of Conditional $\alpha$-Information w.r.t. $\alpha$-Information]
If $Z$ is independent from $(X,Y)$ then $I_\alpha(X;Y|Z)=I_\alpha(X;Y)$. 
\end{property}
\begin{proof}
Obvious from the definitions. 
\end{proof}

\begin{property}[UEP for Conditional $\alpha$-Information]\label{prop-UEP-cond-info}
If $U\sim\mathcal{U}(M)$ is uniformly distributed independent of $Z$, then $I_\alpha(U;Y|Z) = H_\alpha(U) - H_\alpha(U|YZ)=\log M - H_\alpha(U|YZ)$.
\end{property}
\begin{proof}
Similarly as for the preceding UEPs, we have $\langle p_{U|YZ}\|u\rangle_\alpha=M^{\frac{\alpha-1}{\alpha}}\|p_{U|YZ}\|_\alpha$. Averaging over $(Y,Z)$ and taking the logarithm gives the announced formula.
\end{proof}

We say that a sequence of random variables forms a \emph{conditional Markov chain} given some random variable $T$ if it is Markov for any $T=t$. 
\begin{property}[DPI for Conditional $\alpha$-Information]\label{prop-cond-DPI}
If $W-X-Y-Z$ forms a conditional Markov chain given $T$, then $I_\alpha(X;Y|T) \geq I_\alpha(W;Z|T)$.
\end{property}
\begin{proof}
By Property~\ref{prop-DPI-info}, $I_\alpha(X;Y|T=t) \geq I_\alpha(W;Z|T=t)$ for any $t$.
From Definition~\ref{def-info} this gives $\langle p_{X|Y,T}\|p_{X|T}\rangle_\alpha\geq \langle p_{W|Z,T}\|p_{W|T=t}\rangle_\alpha$ for $\alpha>1$ and the opposite inequality for $0<\alpha<1$. This in turn from Definition~\ref{def-cond-info} gives the announced inequality for any $\alpha$.
\end{proof}

\subsection{Conditional Sibson's Identity}

\begin{proposition}[Conditional Sibson's Identity] \label{prop-cond-sibson-identity}
One has
\begin{equation}
D_\alpha(P_{XYZ}\|P_{X|Z}Q_{YZ}) = D_\alpha(Q^*_{YZ}\|Q_{YZ}) + I_\alpha(X;Y|Z), 
\end{equation}
hence the following alternate minimizing definition:
\begin{equation}\label{eq-cond-info-min}
I_\alpha(X;Y|Z)  =\min_{Q_{YZ}}  D_\alpha(P_{XYZ}\|P_{X|Z}Q_{YZ})
\end{equation}
\end{proposition}

\begin{proof}
Similarly as in the case of $\alpha$-information, it is straightforward to compute
\begin{align}
\langle p_{XYZ} \| p_{X|Z} q_{YZ} \rangle^\alpha_\alpha &=
\iiint p^\alpha_{YZ} p^\alpha_{X|YZ} p_{X|Z}^{1-\alpha} q_{YZ}^{1-\alpha} \\
&= \bigl\langle p_{YZ} \langle p_{X|YZ}\|p_{X|Z}\rangle_{\!\alpha} \,\|\, q_{YZ} \bigr\rangle^\alpha_\alpha
\end{align}
Defining the (suitably normalized) distribution  $q^*_{YZ}=p_{YZ} \langle p_{X|YZ}\|p_{X|Z}\rangle_{\!\alpha}  / \E_{YZ} \langle p_{X|YZ}\|p_{X|Z}\rangle_{\!\alpha}$, substituting and taking the logarithm gives the announced identity.
\end{proof}

\section{Comparison to Previous Definitions}\label{sec-comp}

\subsection{Various Other Definitions}

All previous definitions of conditional $\alpha$-information we are aware of are variations of the form~\eqref{eq-cond-info-min} where $\alpha$-divergence is minimized with respect to different probability measures $Q_{X|Z}$, $Q_{Y|Z}$, $Q_Z$ or combinations. There are exactly $2^3=8$ possibilities:
\begin{enumerate}[(i)]%\addtocounter{enumi}{-1}
\item[(o)] $I^{000}_\alpha(X;Y|Z)=\hphantom{\min}D_\alpha(P_{XYZ}\|P_{X|Z}P_{Y|Z}P_Z)$.
\item $I^{001}_\alpha(X;Y|Z)=\min\limits_{Q_Z}D_\alpha(P_{XYZ}\|P_{X|Z}P_{Y|Z}Q_Z)$.
\item $I^{010}_\alpha(X;Y|Z)=\min\limits_{Q_{Y|Z}}D_\alpha(P_{XYZ}\|P_{X|Z}Q_{Y|Z}P_Z)$.
\item $I^{011}_\alpha(X;Y|Z)=\min\limits_{Q_{YZ}}D_\alpha(P_{XYZ}\|P_{X|Z}Q_{YZ})$.
\item $I^{100}_\alpha(X;Y|Z)=\min\limits_{Q_{X|Z}}D_\alpha(P_{XYZ}\|Q_{X|Z}P_{Y|Z}P_Z)$.
\item $I^{101}_\alpha(X;Y|Z)=\min\limits_{Q_{XZ}}D_\alpha(P_{XYZ}\|Q_{XZ}P_{Y|Z})$.
\item $I^{110}_\alpha(X;Y|Z)=\min\limits_{Q_{X|Z}Q_{Y|Z}}D_\alpha(P_{XYZ}\|Q_{X|Z}Q_{Y|Z}P_Z)$.
\item $I^{111}_\alpha(X;Y|Z)=\min\limits_{Q_{X|Z}Q_{YZ}}D_\alpha(P_{XYZ}\|Q_{X|Z}Q_{YZ})$
\end{enumerate}

Definition~(o) is mentionned in~\cite[Eq.~(70)]{TomamichelHayashi18}. Definition (i) is the main proposal of Esposito et al.~\cite{EspositoWuGastpar21}. Definition (ii) is discussed by Tomamichel and Hayashi~\cite[Eq.~(74)]{TomamichelHayashi18} and is equivalent to definition (iv) by permuting the roles of $X$ and $Y$: $I^{100}_\alpha(X;Y|Z)=I^{010}_\alpha(Y;X|Z)$.
Our definition~\eqref{eq-cond-info-min} is definition~(iii), and is equivalent to definition (v) by permuting the roles of $X$ and $Y$: $I^{101}_\alpha(X;Y|Z)=I^{011}_\alpha(Y;X|Z)$.
Finally, definitions~(vi) and~(vii) seem new and related to a conditional version of
the Lapidoth-Pfister mutual information~\cite{LapidothPfister16}:
$J_{\alpha}(X;Y)=\min_{Q_XQ_Y} D_{\alpha} (P_{XY} \parallel Q_{X}Q_{Y})$.
Thus we need only to compare our definition to~(o), (i), (ii), (vi) and~(vii).

We now discuss various properties  for these definitions, by decreasing order of importance: The fact that they admit or not a closed-form expression in terms of the involved probability densities; their consistency with respect to $\alpha$-information $I_\alpha(X;Y|0)=I_\alpha(X;Y)$; the existence of a uniform expansion of the form $I_\alpha(U;Y|Z)=\log M - H_\alpha(U|YZ)$ when  $U\sim \mathcal{U}(M)$ is independent of $Z$; and the fact that they satisfy data processing inequalities for conditional Markov chains. 

\subsection{Closed-Form and Consistency}

Definition~(o) is by itself a closed-form expression but is clearly \emph{inconsistent} with respect to Sibson's $\alpha$-information since $I^{000}_\alpha(X;Y|0)= D_\alpha(P_{XY}\|P_{X}P_{Y})$ which by~\eqref{eq-info-min} is $\geq I_\alpha(X;Y)$ where the inequality is, in general, strict.

Definition~(i) of Esposito et al. does admit a closed-form expression~\cite[Thm.\,2]{EspositoWuGastpar21}.
In fact, since  
\begin{align*}
\langle p_{XYZ} \| p_{X|Z} p_{Y|Z} q_Z \rangle^\alpha_\alpha &=
\!\!\iiint\! p^\alpha_{Z} p^\alpha_{XY|Z} (p_{X|Z}p_{Y|Z})^{1\!-\!\alpha} q_{Y}^{1\!-\!\alpha} \\
&= \bigl\langle p_{Z} \langle p_{XY|Z}\|p_{X|Z}p_{Y|Z}\rangle_{\!\alpha} \,\|\, q_{Z} \bigr\rangle^\alpha_\alpha,
\end{align*}
letting $q^*_{Z}=p_{Z} \langle p_{X|YZ}\|p_{X|Z}p_{Y|Z}\rangle_{\!\alpha}  / \E_{Z} \langle p_{X|YZ}\|p_{X|Z}p_{Y|Z}\rangle_{\!\alpha}$ and taking the logarithm gives the following variation of Sibson's identity (whose existence is mentionned but does not explicitly appear in~\cite{EspositoWuGastpar21}):
\begin{proposition}%[Conditional Sibson's Identity and Closed-Form Expression]
\begin{equation}
D_\alpha(P_{XYZ}\|P_{X|Z}P_{Y|Z}Q_Z) = D_\alpha(Q^*_{Z}\|Q_{Z}) + I^{001}_\alpha(X;Y|Z), 
\end{equation} 
with the following closed-form expression:
\begin{equation}\label{eq-info-esposito}
I_\alpha^{001}(X;Y|Z)=  \tfrac{\alpha}{\alpha-1} \log \E_Z \langle p_{XY|Z}\|p_{X|Z}p_{Y|Z}\rangle_\alpha.
\end{equation}
\end{proposition}
\noindent However, $I_\alpha^{001}$ is \emph{inconsistent} (with respect to Sibson's $\alpha$-information)  for the same reason as in the case of $I^{000}_\alpha$: From~\eqref{eq-info-esposito} we have $I^{001}_\alpha(X;Y|0)= D_\alpha(P_{XY}\|P_{X}P_{Y}) \geq I_\alpha(X;Y)$.

Definition~(ii) of Tomamichel and Hayashi also admits a closed-form expression~\cite[Eq.~(75)]{TomamichelHayashi18}. In fact by the (unconditional) Sibson identity~\eqref{eq-sibson-identity} applied to all variables conditioned on $Z=z$ for any $z$, on easily sees that $D_\alpha(P_{XYZ}\|P_{X|Z}Q_{Y|Z}P_Z)$ achieves its minimum when for $q_{Y|Z}=q^*_{Y|Z}= p_{Y|Z} \langle p_{X|YZ}\|p_X\rangle_{\!\alpha}  / \E_{Y|Z} \langle p_{X|YZ}\|p_{X|Z}\rangle_{\!\alpha}$ as given above in the proof of~\eqref{eq-sibson-identity}, which gives
\begin{equation}\label{eq-info-tomamichel}
I_\alpha^{010}(X;Y|Z)=  \tfrac{1}{\alpha-1} \log \E_Z (\E_{Y|Z}\langle p_{X|YZ}\|p_{X|Z}\rangle_\alpha)^\alpha.
\end{equation}
From this it follows that $I_\alpha^{010}(X;Y|0)=I_\alpha(X;Y)$, proving that $I_\alpha^{010}$ is \emph{consistent}.

Finally, definitions~(vi) and~(vii) are neither closed-form nor consistent; for when $Z\equiv0$, the definitions reduces to the Lapidoth-Pfister mutual information:
$J_{\alpha}(X;Y)=\min_{Q_XQ_Y} D_{\alpha} (P_{XY} \parallel Q_{X}Q_{Y})$
which already does not admit a closed-form expression, and for which $J_{\alpha}(X;Y)\leq I_\alpha(X;Y)$ where the inequality is, in general, strict~\cite{LapidothPfister16}.
In the following we focus on the other definitions which admit closed-form expressions.

\subsection{Uniform Expansion Property}

Using the above closed-form expressions it is easy to check the UEP when $U\sim \mathcal{U}(M)$ is independent of $Z$, neither 
$I^{000}_\alpha(U;Y|Z)$, nor $I^{001}_\alpha(U;Y|Z)$, nor $I^{010}_\alpha(U;Y|Z)$ 
equals $\log M - H_\alpha(U|YZ)$.
This is not surprising since in general, from the different minimizations of $\alpha$-divergence, 
\begin{equation}
\begin{split}
I_\alpha(X;Y|Z)&=I^{011}_\alpha(X;Y|Z) \\&\leq \min \bigl\{ I^{001}_\alpha(X;Y|Z), I^{010}_\alpha(X;Y|Z)\bigr\}
\\&\leq I^{000}_\alpha(X;Y|Z).
\end{split}
\end{equation}
where the inequalities are, in general, strict. Hence the only case where the UEP (which is crucial in our subsequent derivations) holds is for the definition~(iii) proposed in this paper.

\subsection{Data Processing Inequality} 

Finally, since definitions~(o) and~(i) are inconsistent with $I^{000}_\alpha(X;Y|0)=I^{001}_\alpha(X;Y|0)= D_\alpha(P_{XY}\|P_{X}P_{Y})$, they do not even satisfy data processing inequalities for a unconditional Markov chain. Therefore, the only remaining candidate for DPI is definition~(ii).
\begin{property}[DPI for  $I^{010}_\alpha(X;Y|Z)$]
If $W-X-Y-Z$ forms a conditional Markov chain given $T$, then $I_\alpha^{010}(X;Y|T) \geq I_\alpha^{010}(W;Z|T)$.
\end{property}

\begin{proof}
We mirror the proof of Property~\ref{prop-DPI-info}.
Let $P_{X,Y,T}\!\to\!\boxed{P_{X,Z,T|X,Y,T}}\!\to\!P_{X,Z,T}\to\!\boxed{P_{W,Z,T|X,Z,T}}\!\to P_{W,Z,T}$. By the conditional Markov condition, we have $P_{X,Z,T|X,Y,T}=P_{X,T|X,T} P_{Z|X,Y,T}=P_{X,T|X,T} P_{Z|Y,T}$ where $P_{X,T|X,T}$ is the identity operator; similarly $P_{W,Z,T|X,Z,T}=P_{W|X,Z,T}P_{Z,T|Z,T}=P_{W|X,T}P_{Z,T|Z,T}$. Thus if $Q_{Y|T} \to\!\boxed{P_{Z|Y,T}}\!\to Q_{Z|T}$, 
we find $P_{X|T}Q_{Y|T}P_T \to\!\boxed{P_{X,Z,T|X,Y,T}}\!\to P_{X|T}Q_{Z|T}P_T
\to\!\boxed{P_{W,Z,T|X,Z,T}}\!\to P_{W|T}Q_{Z|T}P_T$.
By the data processing inequality for $\alpha$-divergence (Property~\ref{dpi_alpha}),
$D_\alpha(P_{X,Y,T}\|P_{X|T}Q_{Y|T}P_T)\geq D_\alpha(P_{W,Z,T}\|P_{W|T}Q_{Z|T}P_T)\geq I_\alpha(W;Z|T)$.
Minimizing over $Q_{Y|T}$ gives the announced DPI.
\end{proof}

Table~\ref{tab} summarizes the comparison between properties of (o)--(vii).

\begin{table}[h!]
\caption{Comparison of some properties for the various definitions.\label{tab}}
\begin{center}
\begin{tabular}{cccccc}
Definition & Ref. & Closed-form & Consistency & UEP & DPI  \\\hline\hline
o & \cite{TomamichelHayashi18} & yes & no & no & no \\\hline
i & \cite{EspositoWuGastpar21} & yes & no & no & no\\\hline
ii,iv &  \cite{TomamichelHayashi18}& yes & yes & no & yes \\\hline
iii,v & (this paper) & \textbf{yes} & \textbf{yes} &\textbf{yes} & \textbf{yes} \\\hline
vi,vii & --- & no & no & & 
\\\hline
\end{tabular}
\end{center}
\end{table}

\section{Application to Side-Channel Analysis}\label{sec-sca}

\subsection{Theoretical Derivation}

We follow the framework and notations from~\cite{CheriseyGuilleyRioulPiantanida19c,CheriseyGuilleyRioulPiantanida19} and~\cite{HeuserRioulGuilley14}. 
Let $K$ be a secret key and $T$ be a plain text known to the attacker. During cryptographic processing, $X$ is leaked from the implementation and measured as a ``trace'' $Y$ by the attacker at the output of some noisy measurement channel.  The secret key $K$ can take $M$ equiprobable values and is evidently independent of the text $T$.
The leakage function $X=f(K,T)$ is unknown, but deterministic.
The attacker then exploits his knowledge of $T$ and $Y$ to estimate the secret $\hat{K}$ and we let $\P_s=\P(\hat{K}=K)$ be the probability of success. The communication channel model is depicted in Fig.~\ref{fig}.

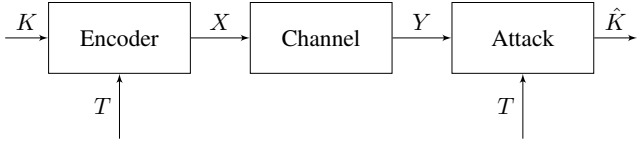
\begin{figure}[h!]
\centering
\resizebox{0.47\textwidth}{!}{
\begin{tikzpicture}[auto, node distance=1.7cm,>=latex']
\node [input, name=input] {};
\node [block, right of=input] (encoder) {Encoder};
\node [block, right of=encoder,  node distance=3cm] (channel) {Channel};
\node [block, right of=channel,  node distance=3cm] (attack) {Attack};
\node [output, right of=attack] (output) {};
\node[input] (text) at ([yshift=-1.5cm]$(encoder)!0.0!(channel)$) {};  
\node[input] (text2) at ([yshift=-1.5cm]$(channel)!1.0!(attack)$) {};  
\draw [->] (encoder) -- node[name=X] {$X$} (channel);
\draw [->] (input) -- node {$K$} (encoder);
\draw [->] (channel)-- node [name=Y] {$Y$}(attack);
\draw [->] (attack) -- node [name=Khat] {$\hat{K}$}(output);
\draw[->] (text) -- node{$T$} (encoder);
\draw[->] (text2) -- node{$T$} (attack);
\end{tikzpicture}}
\caption{Side-channel seen as a communication channel.}\label{fig}
\end{figure}

\begin{theorem} One has the following upper bound on the probability of success $\P_s$:
\label{thm:bound:sr}
\begin{equation}\label{eq-main}
 I_\alpha(X,Y|T) \ge d_{\alpha}(\mathbb{P}_s\parallel \frac{1}{M})
\end{equation}

\end{theorem}

\begin{proof}
The chain $K-X-Y$ is Markov given $T$ by assumption but since $X=f(K,T)$,  the chain $X-K-Y$ is also Markov given $T$. Therefore, by the conditional DPI (Property~\ref{prop-cond-DPI}),  $I_\alpha(X,Y|T)=I_\alpha(K,Y|T)$ (inequalities in both directions).
Now since $K-Y-\hat{K}$ is also Markov given $T$, we have $I_\alpha(K;Y|T)\geq I_\alpha(K;\hat{K}|T)$. Since $K$ is equiprobable independent of $T$, by the UEP (Property~\ref{prop-UEP-cond-info}),
$I_\alpha(K;\hat{K}|T)=\log M - H_\alpha(K|\hat{K},T) \geq \log M - H_\alpha(K|\hat{K}) =  I_\alpha (K;\hat{K})$.
Finally, using Lemma~\ref{lem-gen-fano}, $I_\alpha(K;\hat{K}) \ge  d_{\alpha}(\mathbb{P}_s(K|Y)\parallel \mathbb{P}_s(K) )=d_\alpha(\P_s\|\frac{1}{M})$, which proves~\eqref{eq-main}.
\end{proof}

\begin{remark}
From~\eqref{eq-binary-div}, $d_\alpha(p,q)$ is increasing in $p$ when $p\geq q$. Hence~\eqref{eq-main} gives a upper bound on $\P_s$ (which is obviously $\geq 1/M$ since $\P_s=1/M$ corresponds to a blind guess when the attacker does not know $Y$). 
\end{remark}

\subsection{Numerical Simulations}

We consider an implementation of the AES with a large number $q$ of measurement traces. Here $M=256$ and the most commonly used leakage model is 
\begin{equation}
Y_i = w_H(S(T_i\oplus K)) + N_i \qquad  (i=1,2,\ldots, q)
\end{equation}
where $w_H$ denotes the Hamming weight, $S$ denotes a S-box permutation and $\N_i$ are i.i.d~$\sim\mathcal{N}(0, \sigma^2)$. Letting $\X=(X_i)_i$, $\Y=(Y_i)_i$, $\T=(T_i)_i$, we can compute 
$I_\alpha(\mathbf{X},\mathbf{Y}|\mathbf{T}) =I_\alpha(K,\mathbf{Y}|\mathbf{T})$ using Monte-Carlo simulation similarly as in~\cite{CheriseyGuilleyRioulPiantanida19}.

The numerical results on the success probability \emph{upper bounds} for $\alpha=1/2$, $1$, and $2$.
are shown in Fig.~\ref{fig:att:alpha}, which compares them to the average performance of the optimal ML attack (with error bars). Since $I_\alpha(\mathbf{X},\mathbf{Y}|\mathbf{T})$ increases with $q$, these in turn allows us to derive \emph{lower bounds} on the number of traces $q_{\min}$ which are needed to achieve a given success rate $\P_s$. This is illustrated in Fig.~\ref{fig:att:sr:bound}.

\begin{figure}[!htbp]
	\centering
		\includegraphics[width=0.49\textwidth]{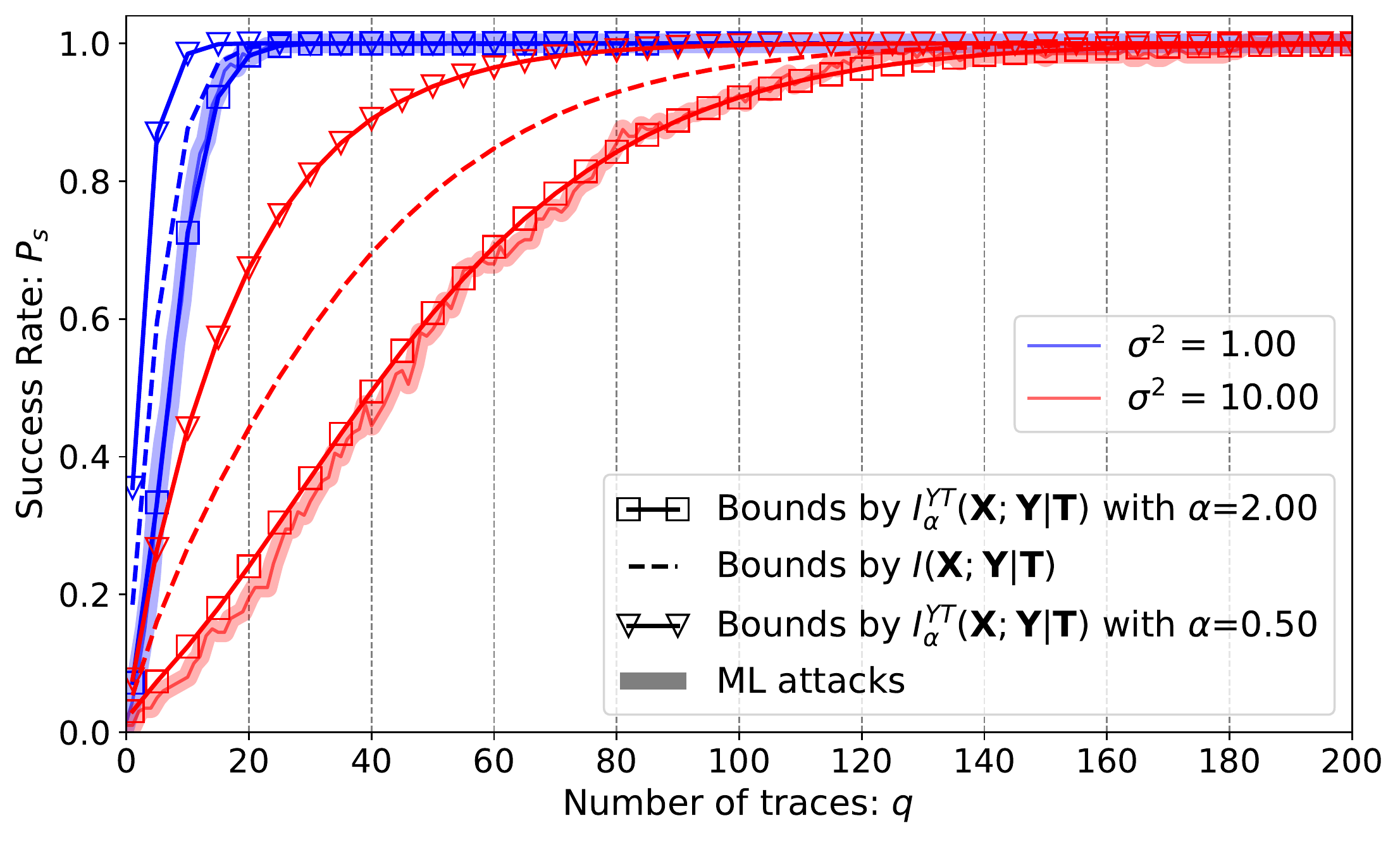}
		\label{fig:att:alpha:2}
	\vspace{-0.4cm}
	\caption{Comparison of upper bounds on success rate $P_s$ given $\alpha$-information $I_\alpha(\mathbf{X},\mathbf{Y}|\mathbf{T})$ for different values of $\alpha$, for a Hamming weight leakage model in an AES-256 implementation.}
	\label{fig:att:alpha}
	\vspace{-0.2cm}
\end{figure}

\begin{figure}[!htbp]
	\centering
	\vspace{-0.2cm}
		\includegraphics[width=0.49\textwidth]{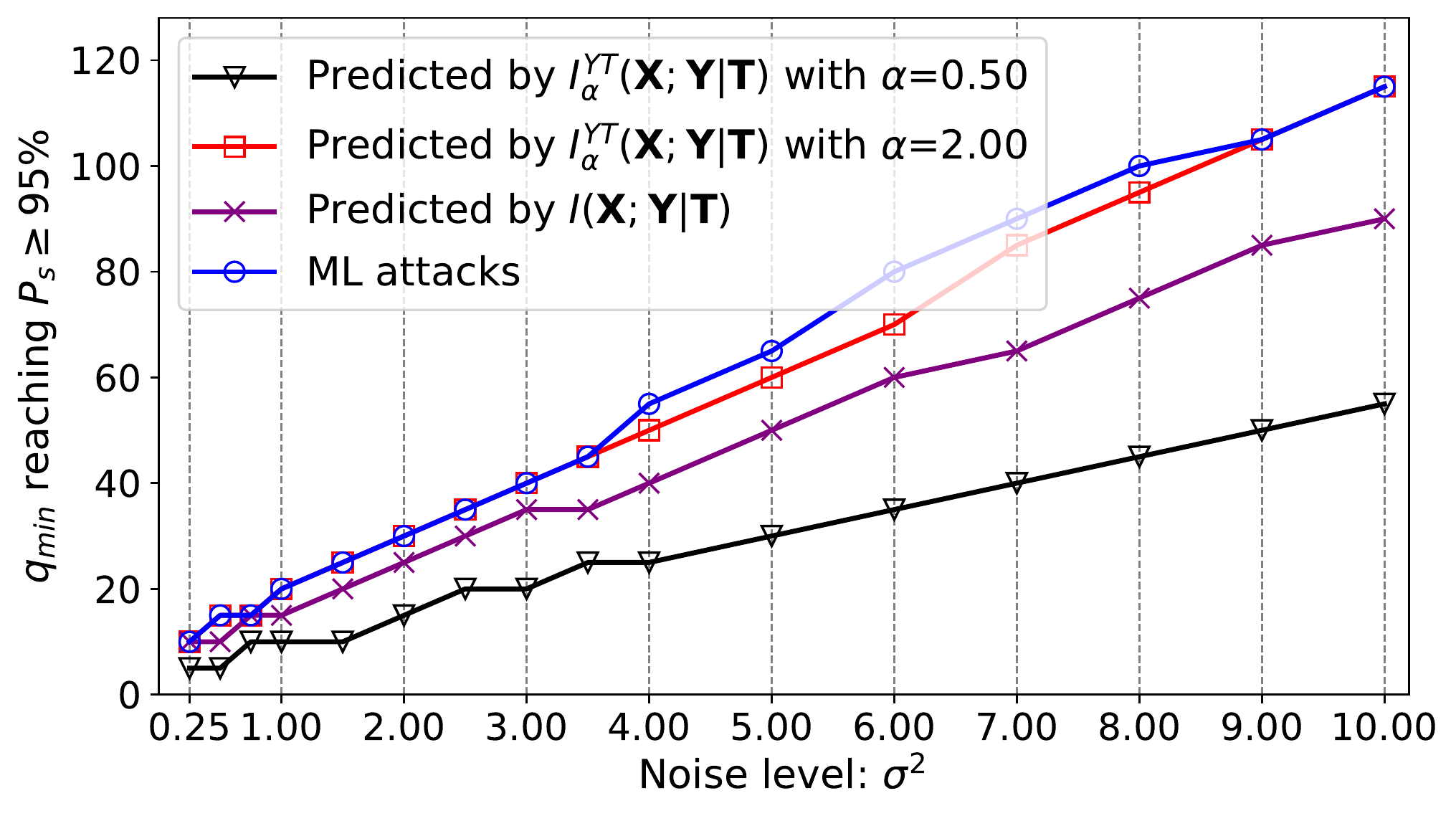}
		\label{fig:att:sr:bound:2}
	\vspace{-0.4cm}
	\caption{Comparison of lower bounds on the number of traces $q_{\min}$ required to reach $\P_s\geq 95\%$ success rate.}
	\label{fig:att:sr:bound}
	\vspace{-0.2cm}
\end{figure}

It is quite remarkable to see that the case $\alpha=2$, corresponding to a \emph{collision} entropy $H_\alpha(K|\hat{K})$, gives a very sharp bound in our setting, which improves the results of~\cite{CheriseyGuilleyRioulPiantanida19c,CheriseyGuilleyRioulPiantanida19}   for $\alpha=1$ very much.
%\section{Conclusion}

\IEEEtriggeratref{9}

\end{document}